\documentclass[numbers,nonatbib,preprint,nocopyrightspace]{sigplanconf}

\newif\ifanonymous
\anonymousfalse

\usepackage{amsmath}
\usepackage{amsthm}
\usepackage{amssymb}
\usepackage{listings}
\usepackage{tikz}
\usepackage{complexity}

\lstdefinelanguage{ML}{
  alsoletter={*},
  morekeywords={datatype, of, if, *},
  sensitive=true,
  morecomment=[s]{/*}{*/},
  morestring=[b]"
}

\lstdefinelanguage{scala}{
  alsoletter={@=>},
  morekeywords={nothing, abstract, case, catch, choose, class, def, do, else, extends, final, finally, for, if, implicit, import, match, new, null, object, let,
override, package, private, protected, requires, return, sealed, super, this, throw, trait, try, type, val, var, while, yield, domain, template, res, time,
postcondition, precondition,invariant, constraint, assert, each, _, return, @generator, ensure, require, ensuring, assuming, otherwise, asserting}
  sensitive=true,
  morecomment=[l]{//},
  morecomment=[s]{/*}{*/},
  morestring=[b]"
}

\newcommand{\codestyle}{\small\sffamily}

 \usepackage{array}
\usepackage{todonotes}
\usepackage{subfig}
\usepackage{float}
\usepackage{stmaryrd}
\usepackage{hyperref}
\usepackage{algorithm}
\usepackage{tikz}
\usetikzlibrary{automata, arrows, calc}

\newcommand{\circq}{\mbox{$\bigcirc$\kern-0.74em{\texttt{?}}\hspace{0.3em}}}

\newcommand{\nonterm}[1]{\ensuremath{\texttt{N}_{\texttt{#1}}}}

\newcommand*{\ruleset}[1]{
	\begin{center}\begin{tabular}{lcl}
			\ruleScan#1\relax\relax
		}
		\newcommand{\ruleScan}[2]{
			\ifx\relax#1
		\end{tabular}\end{center}
		\else
		\nonterm{#1} & $\rightarrow$ & #2\\%
		\expandafter\ruleScan
		\fi
	}

\newcommand{\set}[1]{\{{#1}\}}

\newcommand{\gram}{G}
\newcommand{\NTerm}{N}
\newcommand{\Term}{\Sigma}
\newcommand{\Prod}{R}
\newcommand{\Start}{S}

\newcommand{\rhs}{\mathit{rhs}}
\newcommand{\nterm}{A}
\newcommand{\mkrule}[2]{{#1} \rightarrow {#2}}
\newcommand{\rul}{r}
\newcommand{\west}[1]{\pi({#1})}
\newcommand{\east}[1]{\overline{\pi}({#1})}
\newcommand{\closing}[1]{\overline{#1}}

\newcommand{\final}{\bot}
\newcommand{\getgraph}{{\sf graph}}

\newcommand{\optimalset}{\Phi}

\newcommand{\restrict}[2]{{#1}_{|{#2}}}

\newcommand{\Lin}[1]{\ensuremath{\text{Lin}(#1)}}

\newtheorem{theorem}{Theorem}

\newtheorem{lemma}{Lemma}

\makeatletter
\def\@seccntformat#1{\@ifundefined{#1@cntformat}
   {\csname the#1\endcsname\quad}  
   {\csname #1@cntformat\endcsname}
}
\let\oldappendix\appendix 
\renewcommand\appendix{
    \oldappendix
    \newcommand{\section@cntformat}{\appendixname~\thesection\quad}
}
\makeatother

\begin{document}

\setlength{\pdfpageheight}{\paperheight}
\setlength{\pdfpagewidth}{\paperwidth}

\newcommand{\ourtitle}{Optimal Test Sets for Context-Free Languages}
\titlebanner{}        

\title{\ourtitle}
\subtitle{}

\ifanonymous
\authorinfo{}
          {}
          {}
\else
\authorinfo{Mika\"el Mayer}
          {EPFL}
          {mikael.mayer@epfl.ch}
\authorinfo{Jad Hamza}
          {EPFL/INRIA}
          {jad.hamza@epfl.ch}
\fi

\maketitle

\begin{abstract}

A test set for a formal language (set of strings) L is a subset T of L such that for any two string homomorphisms f and g defined on L, if the restrictions of f and g on T are identical functions, then f and g are identical on the entire L. Previously, it was shown that there are context-free grammars for which smallest test sets are cubic in the size of the grammar, which gives a lower bound on tests set size. Existing upper bounds were higher degree polynomials; we here give the first algorithm to compute test sets of cubic size for all context-free grammars, settling the gap between the upper and lower bound.

\end{abstract} \keywords
test sets, context-free languages, context-free grammars

\section{Introduction}

It is known that given a context-free language $L$ (given by a context-free
grammar $G$ of size $n$), one can construct a test set $T$ for $L$
whose size is $O(n^6)$~
\cite{plandowski_testset_1994,plandowski_testset_1995,plandowski_survey_2003}.

Moreover, it was shown~\cite{plandowski_testset_1994,plandowski_testset_1995,plandowski_survey_2003} that $O(n^3)$ is a lower bound, in the sense that there exists an infinite
family of context-free grammars $G_1,G_2,\dots$, such that the size of $G_n$ is
$O(n)$ and the number of words contained in $G_n$ is 
$O(n^3)$ but $G_n$ does not contain a test set $T$ as a strict subset.
The only test set for $G_n$ is $G_n$.

Our contribution is to prove that the $O(n^3)$ bound is in fact tight.
More specifically, we give an algorithm that given a context-free grammar
$G$ of size $n$, produces a test set $T$ whose size is $O(n^3)$.
We thus greatly improve the original $O(n^6)$ upper bound~
\cite{plandowski_testset_1994,plandowski_testset_1995,plandowski_survey_2003}. \section{Notations and Definitions}

\subsection{Grammars}

A \emph{context-free grammar}
  $\gram$ is a tuple $(\NTerm,\Term,\Prod,\Start)$ where:
\begin{itemize}
\item $\NTerm$ is a set of \emph{non-terminals},
\item $\Term$ is a set of \emph{terminals},
\item $\Prod \subseteq \NTerm \times (\NTerm \uplus \Term)^* $ 
  is a set of \emph{production rules},
\item $\Start \in \NTerm$ is the starting non-terminal symbol.
\end{itemize}

A production $(\nterm,\rhs) \in \Prod$ is denoted $\mkrule{\nterm}{\rhs}$.
The \emph{size} of $\gram$, denoted $|\gram|$, is the sum of sizes of
each production in $\Prod$:
    $\sum_{\mkrule{\nterm}{\rhs} \in \Prod} (|\rhs| + 1)$.

By an abuse of notation, we denote by $\gram$ the set of words
produced by $\gram$.

A grammar is \emph{linear} if for every for every production 
$\mkrule{\nterm}{\rhs} \in \Prod$, the $\rhs$ string contains at most one occurrence
from $\NTerm$.

\subsection{Morphisms and Test Sets}

Given a (partial) function from 
$f: A \to B$, and a set $C$, 
$\restrict{f}{C}$ denotes the (partial) function 
$g: A \cap C \to B$ such that $g(a) = f(a)$ for all
$a \in A \cap C$.

A morphism $f: \Sigma^* \rightarrow \Gamma^*$ is a function
such that $f(\epsilon) = \epsilon$ and for every
$u,v \in \Sigma^*$, $f(u \cdot v) = f(u) \cdot f(v)$,
where the symbol `$\cdot$' denotes the concatenation of words.

A subset $T \subseteq L$ of a language $L$ is a \emph{test set} if for 
any two morphisms $f,g: \Sigma^* \rightarrow \Gamma^*$, 
$\restrict{f}{T} = \restrict{g}{T}$ implies 
$\restrict{f}{L} = \restrict{g}{L}$.

 \section{Test Sets for Context-Free Languages}

\label{sec:testsetscontextfree}
\label{section:cftestset}

\subsection{Plandowski's Test Set}

The following lemma was originally used~\cite{plandowski_testset_1994,plandowski_testset_1995} to show that, 
for any 
linear context-free grammar, there exists a test set containing
at most $O(|\Prod|^6)$ elements.
We show in Section~\ref{subsection:linear}
how this lemma can be used to show a $2|\Prod|^3$ bound.

Let $\Sigma_4 = 
    \set{
        a_i, 
        \closing{a_i},
        b_i,
        \closing{b_i}\ | i \in \set{1,2,3,4}}$
be an alphabet.
We define:
\begin{flalign*}
L_4 = 
    \{&x_4 \, x_3\, x_2 \, x_1 \, 
        \closing{x_1} \, 
        \closing{x_2} \, 
        \closing{x_3} \, 
        \closing{x_4}\ |\ \\
        &\forall i \in \set{1,2,3,4}.\ 
            (x_i,\closing{x_i}) = (a_i,\closing{a_i}) \lor
            (x_i,\closing{x_i}) = (b_i,\closing{b_i})
\}
\end{flalign*}
and 
$T_4 = L_4 \setminus \set{b_4 \,b_3\,b_2\,b_1\,
    \closing{b_1}\,\closing{b_2}\,\closing{b_3}\,\closing{b_4}}$.

The sets $L_4, T_4 \subseteq \Sigma_4$ 
    have $16$ and $15$ elements respectively.

\begin{lemma}[\cite{plandowski_testset_1994,plandowski_testset_1995}]
\label{lemma:t4l4}
$T_4$ is a test set for $L_4$.
\end{lemma}

\subsection{Linear Context-Free Grammars}
\label{subsection:linear}

We now prove that for any context-free grammar $G$,
there exists a test set whose size is $2|\Prod|^3$.
Like the original proof of
\cite{plandowski_testset_1994,plandowski_testset_1995} that gave a 
$O(|\Prod|^6)$ upper bound, our proof 
relies on Lemma~\ref{lemma:t4l4}. However, our proof uses
a different construction to obtain the new, tight, bound.

\begin{theorem}
\label{theorem:cftestset}
Let $\gram = (\NTerm,\Term,\Prod,\Start)$ be a linear context-free grammar. 
There exists a test set $T \subseteq \gram$ for $\gram$ containing 
at most $2|\Prod|^3$ elements.
\end{theorem}

\begin{proof}

Before building the test set, we introduce some notation.

\paragraph{Graph of $\gram$.}

Define the labeled graph $\getgraph(\gram) = (V,E)$ where 
$V = \NTerm \uplus \set{\final}$, 
$E \subseteq V \times \Prod \times V$ such that:
\begin{itemize}
\item 
  for non-terminals $A,B \in \NTerm$ and a rule $\rul \in \Prod$,  let
  $(A,\rul,B) \in E$ iff
  $\rul$ is of the form $\mkrule{A}{u B v}$ where $u,v \in \Sigma^*$ (i.e., $B$ is the only non-terminal occurring in
  $\rhs$).
\item 
  for a non-terminal $A \in \NTerm$ and $\rul \in \Prod$, 
  $(A,\rul,\final) \in E$ if and only if
  $\rul = \mkrule{A}{\rhs}$ for some $\rhs \in \Sigma^*$.
\end{itemize}

A \emph{path} of $\getgraph(\gram)$ is a (possibly cyclic) sequence of 
edges of $E$, of the form:
$(A_1,\rul_1,A_2) \cdot
(A_2,\rul_2,A_3) 
\cdots
(A_n,\rul_n,A_{n+1})$.
A path is \emph{accepting} if $A_1 = \Start$ and $A_{n+1} = \final$.

\paragraph{Link between  $\getgraph(\gram)$ and $\gram$.}

Given a rule $\mkrule{A}{u B v} \in \Prod$,
where $A,B \in \NTerm$ and $u,v \in \Term^*$, 
we denote $\west{\rul} = u$ and $\east{\rul} = v$.
For a rule of the form $\mkrule{A}{u}$ where $u \in \Term^*$
we denote $\west{\rul} = u$ and $\east{\rul} = \epsilon$.
For a path $P = 
(A_1,\rul_1,A_2) \cdot
(A_2,\rul_2,A_3) \cdot
\cdots
(A_n,\rul_n,A_{n+1})$ 
we define $\west{P} = \west{\rul_1} \cdots \west{\rul_n}$,
and $\east{P} = \east{\rul_n} \cdots \east{\rul_1}$.

Each accepting path $P$ in $\getgraph(\gram)$ corresponds to a word 
$\west{P} \cdot \east{P}$ in $\gram$, and
conversely, for any word $w \in \gram$, there exists an
accepting path (not necessarily unique) in $\getgraph(\gram)$
corresponding to $w$.

\paragraph{Total order on paths.}

We fix an arbitrary total order $<$ on $\Prod$, and extend it
to sequence of edges in $\Prod^*$ as follows.
Given paths $P_1,P_2 \in \Prod^*$, we have 
$P_1 < P_2$ iff 
\begin{itemize}
\item $|P_1| < |P_2|$ (length of $P_1$ is smaller than length of $P_2$), or
\item $|P_1| = |P_2|$ and 
  $P_1$ is smaller lexicographically than $P_2$.
\end{itemize}

A path $P$ is called \emph{optimal} if it is the minimal path from the first
vertex of $P$ to the last vertex of $P$.

\paragraph{Test set for $\gram$.}
 
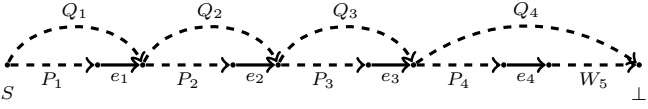
\begin{figure}
    \centering
    \scriptsize
    
    \begin{tikzpicture}[
        pnt/.style={->,circle, fill=black, inner sep=0pt, minimum size=2pt},
        arr/.style={->,dashed, very thick},
        edd/.style={->,very thick},
        inv/.style={draw=none,opacity=0},
        x=0.6cm,
        y=-1cm,
        scale=1
    ]

    \node[pnt] at (0,0) (B0) { };
    \node[pnt] at (2,0) (A1) {  };
    \node[pnt] at (3,0) (B1) {  };
    \node[pnt] at (5,0) (A2) {  };
    \node[pnt] at (6,0) (B2) {  };
    \node[pnt] at (8,0) (A3) {  };
    \node[pnt] at (9,0) (B3) {  };
    \node[pnt] at (11,0) (A4) {  };
    \node[pnt] at (12,0) (B4) {  };
    \node[pnt] at (14,0) (A5) {  };
    
    \path[arr] (B0) edge[bend left=65] node[above] { $Q_1$ } (B1);
    \path[arr] (B1) edge[bend left=65] node[above] { $Q_2$ } (B2);
    \path[arr] (B2) edge[bend left=65] node[above] { $Q_3$ } (B3);
    \path[arr] (B3) edge[bend left=35] node[above] { $Q_4$ } (A5);
    
    \path[arr] (B0) edge[bend left=0] node[below] { $P_1$ } (A1);
    \path[arr] (B1) edge[bend left=0] node[below] { $P_2$ } (A2);
    \path[arr] (B2) edge[bend left=0] node[below] { $P_3$ } (A3);
    \path[arr] (B3) edge[bend left=0] node[below] { $P_4$ } (A4);
    \path[arr] (B4) edge[bend left=0] node[below] { $W_5$ } (A5);

    \path[edd] (A1) edge[bend left=0] node[below] { $e_1$ } (B1);
    \path[edd] (A2) edge[bend left=0] node[below] { $e_2$ } (B2);
    \path[edd] (A3) edge[bend left=0] node[below] { $e_3$ } (B3);
    \path[edd] (A4) edge[bend left=0] node[below] { $e_4$ } (B4);
    
    \node[below=2mm] at (B0) { $\Start$ };
    \node[below=2mm] at (A5) { $\final$ };
    
    \end{tikzpicture}
    \caption{
        The four optimal subpaths $Q_1$, $Q_2$, $Q_3$, and $Q_4$ define
        $15$ alternative paths from $\Start$ to $\final$ which are all
        strictly smaller (with respect to order $<$) than
        $P_1 e_1 P_2 e_2 P_3 e_3 P_4 e_4 W_5$.
    }
    \label{figure:cftestset}
\end{figure}

Let $\optimalset_k(\gram)$ be the set of words of $\gram$ corresponding to accepting
paths of the form
$P_1 e_1 P_2 \cdots P_n e_n P_{n+1}$, $n \leq k$,
with $P_i \in \Prod^*$, $e_i \in \Prod$, and 
where for $i \in \set{1,\dots,n+1}$, 
$P_i$ is optimal,
and for $i \in \set{1,\dots,n}$,
$P_i e_i$ is not optimal.
By construction, a path in $\optimalset_k(\gram)$ is uniquely
determined (when it exists) by the choice 
of edges $e_1,\dots,e_n$, as optimal paths between
two vertices are unique.
Therefore, $\optimalset_k(\gram)$
contains at most $\sum_{i = 0}^k |\Prod|^i \leq 
2|\Prod|^k$ words.

We now show that $\optimalset_3(\gram)$ is a test set for $\gram$
(which gives us the desired bound of the theorem: $2|\Prod|^k$).
Assume there exist two morphisms $f, g: \Term^* \to \Gamma^*$ such that 
$\restrict{f}{\optimalset_3(\gram)} = \restrict{g}{\optimalset_3(\gram)}$ and 
there exists $w \in G$ such that $f(w) \neq g(w)$.

By assumption, $w$ does not belong to $\optimalset_3(\gram)$, and must correspond 
to a path $P = P_1 e_1 P_2 \cdots P_n e_n P_{n+1}$ for $n \geq 4$, such that
for $i \in \set{1,\dots,n+1}$, 
$P_i$ is optimal, and $P_i e_i$ is not optimal.
We pick $w$ having the property $f(w) \neq g(w)$ such that the path $P$ is the 
smallest possible (according to the order $<$ defined above).

The path $P$ can be 
written $P_1 e_1 P_2 e_2 P_3 e_3 P_4 e_4 W_5$ where 
for $i \in \set{1,2,3,4}$, $P_i$ is optimal, and 
$P_i e_i$ is not optimal ($W_5$ is not necessarily optimal).
For $i \in \set{1,2,3}$,
we define $Q_i$ to be the optimal path from the source of
$P_i e_i$ to its target; hence $Q_i < P_i e_i$.
Moreover, $Q_4$ is defined to be the optimal path
from the source of $P_4 e_4 W_5$ to its target, with $Q_4 < P_4 e_4 W_5$.
Effectively, as shown in Figure~\ref{figure:cftestset},
this defines $15$ paths that can be derived from $P$
by replacing subpaths by their corresponding optimal path 
($Q_1$, $Q_2$, $Q_3$, $Q_4$).

Let $P'$ be one of those $15$ paths (where at least one subpath 
has been replaced by its optimal counterpart
$Q_1$, $Q_2$, $Q_3$, or $Q_4$), and let $w' \in \gram$
be the word corresponding to $P'$.
By construction of $P'$, and by definition of the order $<$, 
we have $P' < P$.
Since we have chosen $P$ to be the optimal path such that $f$ and $g$
are not equal on the corresponding word, we deduce that
$f(w') = g(w')$.

To conclude, 
we show that we obtain a contradiction, thanks to 
Lemma~\ref{lemma:t4l4}.
For this, 
we construct two morphisms $f', g': \Sigma_4 \to \Gamma$ as follows
($i$ ranges over $\set{1,2,3,4}$ and $j$ over $\set{1,2,3}$):
\begin{itemize}
\item $f'(a_i) = f(\west{Q_i})$,
\item $f'(\closing{a_i}) = f(\east{Q_i})$,
\item $f'(b_j) = f(\west{P_j e_j})$,
\item $f'(\closing{b_j}) = f(\east{P_j e_j})$.
\item $f'(b_4) = f(\west{P_4 e_4 W_5})$,
\item $f'(\closing{b_4}) = f(\east{P_4 e_4 W_5})$.
\end{itemize}
The morphism $g'$ is defined similarly, using $g$ instead of $f$.
We can then verify that $f'$ and $g'$
coincide on $T_4$, but are not equal on the word 
$b_4 \,b_3\,b_2\,b_1\,
    \closing{b_1}\,\closing{b_2}\,\closing{b_3}\,\closing{b_4}
    \in L_4$, thus contradicting Lemma~\ref{lemma:t4l4}.
\end{proof}

\subsection{Context-Free Grammars}

To obtain a test set for a context-free grammar $\gram$ which
is not necessarily linear, 
\cite{plandowski_testset_1994} constructs from $\gram$ a
linear context-free grammar $\Lin{\gram}$ which produces a subset of $\gram$
which is a test set for $\gram$.

Formally, $\Lin{\gram}$ is derived from $\gram$ as follows:
\begin{itemize}
\item 
    For every productive non-terminal symbol $A$ in \gram{},
    we choose a word $x_A$ that is produced by $A$.
\item 
    Every rule $r: A \to x_0 A_1 x_1 \ldots A_n x_n$ in $\gram$, where for
    every $i$, 
    $x_i \in \Term^*$ and $A_i \in \NTerm$ is productive, 
    is replaced
    by $n$ different rules, each one obtained from $r$ by 
    replacing all $A_i$ with $x_{A_i}$ except one.
\end{itemize}

Note that the definition of \Lin{\gram} is not unique, and depends on the 
choice of the words $x_A$. The following result holds  for any 
choice of the words $x_A$.

\begin{lemma}[\cite{plandowski_testset_1994,plandowski_testset_1995}]
\label{lemma:ling}
$\Lin{\gram}$ is a test set for $\gram$.
\end{lemma}

Using Theorem~\ref{theorem:cftestset}, we improve the 
$O(|G|^6)$ bound of \cite{plandowski_testset_1994,plandowski_testset_1995}
for the test set of $\gram$ to $2|\gram|^3$.

\begin{theorem}
\label{theorem:cftestsetgen}
Let $\gram = (\NTerm,\Term,\Prod,\Start)$ be a context-free grammar. 
There exists a test set $T \subseteq \gram$ for $\gram$ containing 
at most $2|\gram|^3$ elements.
\end{theorem}

\begin{proof}
Follows from Theorem~\ref{theorem:cftestset},
Lemma~\ref{lemma:ling}, and from the fact that 
$\Lin{\gram}$ has at most
$|\gram| = \sum_{\mkrule{\nterm}{\rhs} \in \Prod} (|\rhs| + 1)$
rules. (When constructing $\Lin{\gram}$, each rule 
$\mkrule{\nterm}{\rhs}$ of $\gram$ is duplicated at most 
$|\rhs|$ times.)
\end{proof}

\subsection{Construction of $\optimalset_3(\gram)$}

To construct $\optimalset_3(\gram)$ for a linear context-free 
grammar $\gram = (\NTerm,\Term,\Prod,\Start)$,
we precompute in time
$O(|\NTerm|^2 |\Prod|)$, for each pair of vertices 
$(A,B)$, the optimal path from $A$ to $B$ in 
$\getgraph(\gram)$.
Then for each possible choice of at most $3$ edges 
$e_1 = (A_1,r_1,B_1)$,
\dots
$e_n = (A_n,r_n,B_n)$,
with $0 \leq n \leq 3$, 
we construct the path 
$P = P_1 e_1 \dots P_n e_n P_{n+1}$ where 
each $P_i$ is the optimal path from $A_{i-1}$ to $B_i$
(if it exists) with $A_0 = \Start$ and $B_{n+1} = \final$
by convention.
We then add 
the word corresponding to $P$ to our result.

To conclude, since the length of each optimal path is bounded by
$|\NTerm|$, we can construct 
$\optimalset_3(\gram)$ in time $O(|\NTerm|  \cdot |\Prod|^3)$.

\section{Acknowledgements}

Thanks to Viktor Kuncak, Mukund Raghothaman, and 
Ravichandhran Madhavan for the helpful talks.

{
\bibliographystyle{abbrv}
\softraggedright
\bibliography{main}
}

\appendix

\end{document}